\documentclass{elsart}

\usepackage{graphicx}
\usepackage{amsmath}
\usepackage{amssymb}
\usepackage{amsthm}
\usepackage{mathtools}
\usepackage{color}

\newcommand{\no}[1]{}
\renewcommand{\log}{\lg}
\newcommand{\Bleft}{B_\mathit{left}}
\newcommand{\Bright}{B_\mathit{right}}

\newcommand{\bigO}{\mathcal{O}}

\newtheorem{theorem}{Theorem}
\newtheorem{lemma}{Lemma}
\theoremstyle{definition}
\newtheorem{definition}{Definition}

\begin{document}

\begin{frontmatter}

\title{Universal Compressed Text Indexing\thanksref{Shonan}}

\author{%
Gonzalo Navarro\thanksref{Chile}}

\address{
Center for Biotechnology and Bioengineering (CeBiB), \\
Department of Computer Science, University of Chile. \\
{\tt gnavarro@dcc.uchile.cl}}

\thanks[Shonan]{This work started during Shonan Meeting 126 ``Computation over Compressed Structured Data''.}

\thanks[Chile]{Partially funded by Fondecyt Grant 1-170048 and Basal Funds
FB0001, Chile.}

\author{Nicola Prezza\thanksref{Pisa}}

\address{
Department of Computer Science, University of Pisa, Italy. \\
{\tt nicola.prezza@di.unipi.it}
\vspace*{-5mm}}

\thanks[Pisa]{This work was partially done while the author was holding a
post-doc position at the Technical University of Denmark (DTU). Partially
funded by the project MIUR-SIR CMACBioSeq (``Combinatorial methods for
analysis and compression of biological sequences'') grant n.~RBSI146R5L.}

\begin{abstract}
	
	The rise of repetitive datasets has lately generated a lot of interest
in compressed self-indexes based on dictionary compression, a rich and
heterogeneous family of techniques that exploits text repetitions in different
ways. For each such compression scheme, several different indexing solutions
have been proposed in the last two decades. To date, the fastest indexes for
repetitive texts are based on the run-length compressed Burrows--Wheeler
transform (BWT) and on the Compact Directed Acyclic Word Graph (CDAWG). The
most space-efficient indexes, on the other hand, are based on the Lempel--Ziv
parsing and on grammar compression. Indexes for more universal schemes such as
collage systems and macro schemes have not yet been proposed. Very recently,
Kempa and Prezza [STOC 2018] showed that all dictionary compressors can be
interpreted as approximation algorithms for the smallest \emph{string
attractor}, that is, a set of text positions capturing all distinct
substrings. Starting from this observation, in this paper we develop the first
\emph{universal} compressed self-index, that is, the first indexing data
structure based on string attractors, which can therefore be built on top of
any dictionary-compressed text representation. Let $\gamma$ be the size of a
string attractor for a text of length $n$. From known reductions, $\gamma$ can
be chosen to be asymptotically equal to any repetitiveness measure: number of
runs in the BWT, size of the CDAWG, number of Lempel--Ziv phrases, number of
rules in a grammar or collage system, size of a macro scheme. Our index takes
$\bigO(\gamma\log(n/\gamma))$ words of space and supports locating the $occ$
occurrences of any pattern of length $m$ in $\bigO(m\log n +
occ\log^{\epsilon}n)$ time, for any constant $\epsilon>0$. This is, in
particular, the first index for general macro schemes and collage systems. Our
result shows that the relation between indexing and compression is much deeper
than what was previously thought: the simple property standing at the core of
all dictionary compressors is sufficient to support fast indexed queries.

\begin{keyword}
Repetitive sequences; Compressed indexes; String attractors
\end{keyword}

\end{abstract}

\end{frontmatter}

\section{Introduction}

Efficiently indexing repetitive text collections is becoming of great importance due to the accelerating rate at which repetitive datasets are being produced in domains such as biology (where the number of sequenced individual genomes is increasing at an accelerating pace) and the web (with databases such as Wikipedia and GitHub being updated daily by thousands of users). 
A self-index on a string $S$ is a data structure that 
offers direct access to any substring of $S$ (and thus it replaces $S$), and 
at the same time supports indexed queries such as counting and locating pattern occurrences in $S$. 
Unfortunately, classic self-indexes~--- for example, the FM-index~\cite{FM05}~--- that work extremely well on standard datasets  fail on repetitive collections in the sense that their compression rate does not reflect the input's information content. This phenomenon can be explained in theory with the  fact that entropy compression is not able to take advantage of repetitions longer than (roughly) the logarithm of the input's length~\cite{gagie2006large}. For this reason, research in the last two decades focused on self-indexing based on  \emph{dictionary compressors} such as the Lempel--Ziv 1977 factorization (LZ77)~\cite{lempel1976complexity}, the run-length encoded Burrows--Wheeler transform (RLBWT)~\cite{burrows1994block} and context-free grammars (CFGs)~\cite{KiefferY00}, just to name the most popular ones. The  idea underlying  these compression techniques is to break the text into phrases coming from a dictionary (hence the name \emph{dictionary compressors}), and to represent each phrase using limited information (typically, a pointer to other text locations or to an external set of strings). This scheme allows taking full advantage of long repetitions; as a result, dictionary-compressed self-indexes can be orders of magnitude more space-efficient than entropy-compressed ones on repetitive datasets. 

The landscape of indexes for repetitive collections reflects that of dictionary compression strategies, with specific indexes developed for each compression strategy. Yet, a few main techniques stand at the core of most of the indexes. To date, the fastest indexes are based on the RLBWT and on the Compact Directed Acyclic Word Graph (CDAWG)~\cite{blumer1987complete,crochemore1997direct}. These indexes achieve optimal-time queries (i.e., asymptotically equal to those of suffix trees~\cite{Weiner73}) at the price of a space consumption higher than that of other compressed indexes. Namely, the former index~\cite{GNP18} requires $\bigO(r\log(n/r))$ words of space, $r$ being the number of equal-letter runs in the BWT of $S$, while the latter~\cite{BCspire17} uses $\bigO(e)$ words, $e$ being the size of the CDAWG of $S$. These two measures (especially $e$) have been experimentally confirmed to be not as small as others~--- such as the size of LZ77~--- on repetitive collections~\cite{belazzougui2015composite}. 

Better measures of repetitiveness include the size $z$ of the LZ77 factorization of $S$, the minimum size $g$ of a CFG (i.e., sum of the lengths of the right-hand sides of the rules) generating $S$, or the minimum size $g_{rl}$ of a run-length CFG~\cite{NIIBT16} generating $S$. Indexes using $\bigO(z)$ or $\bigO(g)$ space do exist, but optimal-time queries have not yet been achieved within this space. Kreft and Navarro \cite{KN13} introduced a self-index based on LZ77 
compression, which proved to be extremely space-efficient on highly repetitive
text collections \cite{CFMN16}. Their self-index uses $\bigO(z)$ space
and finds all the $occ$ occurrences of a pattern of length $m$ in time 
$\bigO((m^2 h + (m+occ)\log z)\log(n/z))$, where $h \le z$ is the maximum number of 
times a symbol is successively copied along the LZ77 parsing. A 
string of length $\ell$ is extracted in $\bigO(h\ell)$ time.
Similarly, self-indexes of size $\bigO(g)$ building on grammar compression \cite{CNfi10,CNspire12} can locate all $occ$ occurrences of a pattern in
$\bigO(m^2\log\log n + m\lg z+occ\lg z)$ time. Within this space, a
string of length $\ell$ can be extracted in time 
$\bigO(\log n + \ell/\log_\sigma n)$ \cite{BPT15}.
Alternative strategies based on Block Trees (BTs)~\cite{BGGKOPT15}  appeared recently.
A BT on $S$ uses $\bigO(z\log(n/z))$
space, which is also the best asymptotic space obtained with grammar compressors
\cite{CLLPPSS05,Sak05,Jez15,Jez16,Ryt03}.
In exchange for using more space than LZ77 compression, the BT offers fast
extraction of substrings: $\bigO((1+\ell/\log_\sigma n)\log(n/z))$ time.
A self-index based on BTs has recently been described by 
Navarro \cite{Navspire17}. Various indexes based on combinations of LZ77, CFGs,
and RLBWTs have also been proposed~\cite{GGKNP12,GGKNP14,belazzougui2015composite,NIIBT16.2,BEGV18,CE18}. Some of their best results are
$\bigO(z\log(n/z) +z\log\log z)$ space with 
$\bigO(m+occ (\log\log n+\log^\epsilon z))$ query time \cite{CE18}, and
$\bigO(z\log(n/z))$ space with either 
$\bigO(m\log m+occ\log\log n)$ \cite{GGKNP14} or 
$\bigO(m+\log^\epsilon z + occ(\log\log n + \log^\epsilon z))$ 
\cite{CE18} query time. Gagie et al.~\cite{GNP18} give a more detailed survey.

The above-discussed compression schemes are the most popular, but not the most
space-efficient. More powerful compressors (NP-complete to optimize) include macro schemes~\cite{storer1982data} and collage systems~\cite{KidaMSTSA03}. Not much work exists in this direction, and no indexes are known for these particular compressors. 

\subsection{String attractors}

As seen in the previous paragraphs, the landscape of self-indexes based on dictionary compression~--- as well as that of dictionary compressors themselves~--- is extremely fragmented, with several techniques being developed for each distinct compression strategy. Very recently, Kempa and Prezza \cite{kempa2018roots} gathered all dictionary compression techniques under a common theory: they showed that  these algorithms are approximations to the smallest \emph{string attractor}, that is, a set of text positions ``capturing'' all distinct substrings of $S$.

\bigskip
\begin{definition}[String attractor~\cite{kempa2018roots}]\label{def: string attractor}
	A \emph{string attractor} of a string $S[1..n]$ is a set of $\gamma$ positions $\Gamma = \{j_1, \dots, j_\gamma\}$ such that every substring $S[i..j]$ has an occurrence $S[i'..j'] = S[i..j]$ with $j_k \in [i',j']$, for some $j_k\in\Gamma$. 
\end{definition}

Their main result is a set of reductions from dictionary compressors to string 
attractors of asymptotically the same size. 

\bigskip
\begin{theorem}[\cite{kempa2018roots}]\label{th:reductions}
	Let $S$ be a string and let $\alpha$ be any of these measures: 
	\begin{enumerate}
		\item[(1)] the size $g$ of a CFG for $S$,
		\item[(2)] the size $g_{rl}$ of a run-length CFG for $S$,
		\item[(3)] the size $c$ of a collage system for $S$,
		\item[(4)] the size $z$ of the LZ77 parse of $S$,
		\item[(5)] the size $b$ of a macro scheme for $S$.
	\end{enumerate}
	Then, $S$ has a string attractor of size $\gamma=\bigO(\alpha)$. In all cases, the corresponding attractor can  be computed in $\bigO(|S|)$ time and space from the compressed representation.
\end{theorem}

Importantly, this implies that any data structure based on string attractors is \emph{universal}: given any dictionary-compressed text representation, we can induce a string attractor and build the data structure on top of it. Indeed, the authors  exploit this observation and provide the first universal data structure for random access, of size $\bigO(\gamma\log(n/\gamma))$. Their extraction time within this space is $\bigO(\log(n/\gamma)+\ell/\log_\sigma n)$. By using slightly more space, $\bigO(\gamma\log(n/\gamma)\log^\epsilon n)$ for any constant $\epsilon>0$, they obtain time $\bigO(\log(n/\gamma)/\log\log n + \ell/\log_\sigma n)$, which is the optimal that can be reached using any space in $\bigO(\gamma\,\textrm{polylog}\, n)$ \cite{kempa2018roots}. 
This suggests that compressed computation can be performed independently from the compression method used while at the same time matching the lower bounds of individual compressors (at least for some queries such as random access). 

\subsection{Our Contributions}

In this paper we exploit the above observation and describe {\em the first
universal self-index} based on string attractors, that is, the first indexing
strategy not depending on the underlying compression scheme. Since string
attractors stand at the core of the notion of compression, our result shows
that the relation between compression and indexing is much deeper than what
was previously thought: the simple string attractor property introduced in
Definition~\ref{def: string attractor} is sufficient to support indexed pattern searches.

\bigskip
\begin{theorem} \label{thm:main}
Let a string $S[1..n]$ have an attractor of size $\gamma$. Then, for any 
constant $\epsilon>0$, there exists a data structure of size 
$\bigO(\gamma \log(n/\gamma))$ that, given a pattern string $P[1..m]$, 
outputs all the $occ$ occurrences of $P$ in $S$ in time 
$\bigO(m\log n + occ (\log^\epsilon \gamma + \log\log(n/\gamma))) =
\bigO(m\log n + occ \log^\epsilon n)$.
It can be built in $\bigO(n+\gamma\lg(n/\gamma)\sqrt{\log \gamma})$ worst-case time and $\bigO(\gamma \log(n/\gamma))$ space with a Monte Carlo method returning a correct result with high probability. 
A guaranteed construction, using Las Vegas randomization, takes $\bigO(n\log n)$
expected time (this time also holds w.h.p.) and $\bigO(n)$ space.
\end{theorem}

We remark that no representation offering random access within $o(\gamma\log(n/\gamma))$ space is known.
The performance of our index is close to that of the fastest self-indexes built on other repetitiveness masures, and it is the first one that can be built for macro schemes and collage systems. 

To obtain our results, we adapt the block tree index of 
Navarro~\cite{Navspire17}, which is designed for block trees on the LZ77
parse, to operate on string attractors. The result is also different from the 
block-tree-like structure Kempa and Prezza use for extraction 
\cite{kempa2018roots}, because that one is aligned
with the attractors and this turns out to be unsuitable for indexing. Instead,
we use a block-tree-like structure, which we dub $\Gamma$-tree, which partitions
the text in a regular form. We moreover introduce recent techniques 
\cite{GGKNP14} to remove the quadratic dependency on the pattern length in query
times.

\subsection{Notation}

We denote by $S[1..n] = S[1] \cdots S[n]$ a string of length $n$ over an
alphabet of size $\sigma = \bigO(n)$. Substrings of $S$ are denoted $S[i..j] =
S[i] \cdots S[j]$, and they are called prefixes of $S$ if $i=1$ and suffixes of
$S$ if $j=n$. The concatenation of strings $S$ and $S'$ is denoted
$S \cdot S'$.
We assume the RAM model of computation with a computer word of $\omega =
\Omega(\log n)$ bits. By $\lg$ we denote the logarithm function, to the base 2
when this matters.

The term \emph{with high probability} (\emph{w.h.p.} abbreviated) indicates with probability at least $1-n^{-c}$ for an arbitrarily large constant $c$, where $n$ is the input size (in our case, the input string length).

In our results, we make use of a modern variant of Karp--Rabin fingerprinting \cite{KR87} (more common nowadays than the original version), defined as follows. 
Let $q \ge \sigma$ be a prime number, and $r$ be a uniform number in $[1..q-1]$. The
fingerprint $\hat \phi$ of a string $S = S[1] \cdots S[n]\in [1..\sigma]^n$ is defined as $\hat\phi(S) = \sum_{i=0}^{n-1}S[n-i]\cdot r^i \mod q$. The \emph{extended fingerprint} of $S$ is the triple $\phi(S) = \langle \hat\phi(S), r^{|S|}\mod q, r^{-|S|}\mod q \rangle$. 
We say that $S\neq S'$, with $S,S'\in [1..\sigma]^n$, collide through $\phi$
(for our purposes, it will be sufficient to consider equal-length strings) if
$\hat\phi(S)=\hat\phi(S')$, that is, $\hat\phi(S-S')=0$, where $S''=S-S'$ is the string defined as $S''[i]=S[i]-S'[i] \mod q$.
Since $\hat\phi(S'')$ is a polynomial (in the variable $r$) of degree at most $n-1$ in the field
$\mathcal Z_q$, it has at most $n-1$ roots. As a consequence,
the probability of having a collision between two strings is bounded by
$\bigO(n/q)$ when $r$ is uniformly chosen in $[1..q-1]$. By choosing $q \in
\Theta(n^{c+2})$ for an arbitrarily large constant $c$, one obtains that such
a  hash function is collision-free among all equal-length substrings of a
given string $S$ of length $n$ w.h.p. To conclude, we will exploit the (easily provable) folklore fact that two extended fingerprints $\phi(U)$ and $\phi(V)$ can be combined in constant time to obtain the extended fingerprint $\phi(U V)$. Similarly, $\phi(U V)$ and $\phi(U)$ (respectively, $\phi(V)$) can be combined in constant time to obtain $\phi(V)$ (respectively, $\phi(U)$).
From now on, we will by default use the term ``fingerprint'' to mean extended
fingerprint.

\section{$\Gamma$-Trees}\label{sec:gamma-tree}

Given a string $S[1..n]$ over an alphabet $[1..\sigma]$ with an attractor
$\Gamma$ of size $\gamma$, we define a $\Gamma$-tree on $S$ as follows.
At the top level, numbered $l=0$, we split $S$ into $\gamma$ substrings
(which we call blocks) of length $b_0=n/\gamma$. Each block is then recursively
split into two so that if $b_l$ is the length of the blocks at level $l$, then it 
holds that $b_{l+1}=b_l/2$, until reaching blocks of one symbol after $\lg(n/\gamma)$
levels (that is, $b_l = n/(\gamma \cdot 2^l)$).\footnote{For simplicity of description, we assume that $n/\gamma$ is a power of $2$.} At each level $l$, every block that is at distance $<b_l$ from a
position $j \in \Gamma$ is {\em marked} (the distance between $j$ and a block 
$S[i..i']$ is $i-j$ if $i>j$, $j-i'$ if $i'<j$, and $0$ otherwise). 
Blocks $S[i..i']$ that are not marked are replaced by a pointer $\langle
ptr_1,ptr_2,\delta \rangle$ to an occurrence $S[j'..j'']$ of $S[i..i']$ that 
includes a position $j \in \Gamma$, $j' \le j \le j''$. Such an occurrence 
exists by Definition~\ref{def: string attractor}. Moreover, it must be covered by 1 or 2 consecutive 
marked blocks of the same level due to our marking mechanism, because all the 
positions in $S[j'..j'']$ are at distance $<b_l$ from $j$. Those 1 or 2 nodes of 
the $\Gamma$-tree are $ptr_1$ and $ptr_2$, and $\delta$ is the offset of $j'$
within $ptr_1$ ($\delta=0$ if $j'$ is the first symbol inside $ptr_1$).

In level $l+1$ we explicitly store only the children of the blocks that were marked in 
level $l$. The blocks stored in the $\Gamma$-tree (i.e., all blocks at level 0 and those having a marked parent) are called \emph{explicit}.
In the last level, the marked blocks store their corresponding single symbol
from $S$. 

\begin{figure}[t]
	\centering
	\includegraphics[scale=0.6, trim={0 16cm 5cm 0},clip]{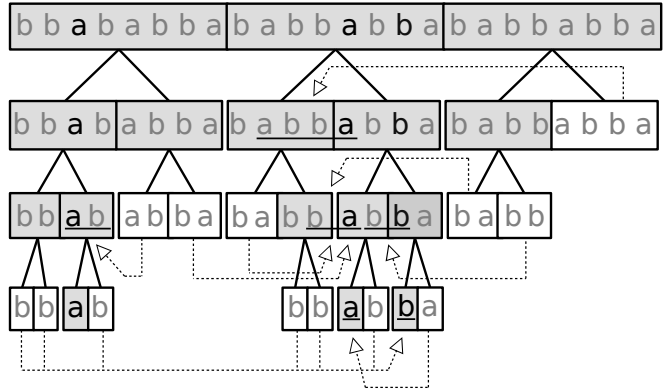}\caption{Example of a $\Gamma$-tree built on a text of length $n=24$ with $\gamma=3$ attractor positions (black letters). Marked blocks are colored in gray. Each non-marked block (in white) is associated with an occurrence (underlined) crossing an attractor position, and therefore overlapping only marked blocks. Only explicit blocks are shown.}\label{fig:gamma tree}
\end{figure}

See Figure \ref{fig:gamma tree} for an example of a $\Gamma$-tree.
We can regard the $\Gamma$-tree as a binary tree (with the first $\lg \gamma$ 
levels chopped out), where the internal nodes are marked nodes and have two
children, whereas the leaves either are marked and represent just one symbol,
or are unmarked and represent a pointer to 1 or 2 marked nodes of the same 
level. If we call $w$ the number of leaves, then there are $w-\gamma$ (marked) 
internal nodes. From the leaves, $\gamma$ of them represent single 
symbols in the last level, while the other $w-\gamma$ leaves are unmarked 
blocks replaced by pointers. Thus, there are in total $2w-\gamma$ nodes in the
tree, of which $w-\gamma$ are internal nodes, $w-\gamma$ are pointers, and 
$\gamma$ store explicit symbols. Alternatively, $w$ nodes are marked (internal
nodes plus leaves) and $w-\gamma$ are unmarked (leaves only).

The $\Gamma$-tree then uses $\bigO(w)$ space. To obtain a bound in terms of 
$n$ and $\gamma$, note that,
at each level, each $j \in \Gamma$ may mark up to $3$ blocks; thus there 
are $w \le 3\gamma\lg(n/\gamma)$ marked blocks in total and the 
$\Gamma$-tree uses $\bigO(\gamma\lg(n/\gamma))$ space.

We now describe two operations on $\Gamma$-trees that are fundamental to
support efficient indexed searches. The former is also necessary for a 
self-index, as it allows us extracting arbitrary substrings of $S$ efficiently.
We remind that this procedure is not the same described on the original 
structure of string attractors~\cite{kempa2018roots}, because the structures are also different.

\subsection{Extraction} \label{sec:extract}

To extract a single symbol $S[i]$, we first map it to a local offset $1 \le
i' \le b_0$ in its corresponding level-$0$ block. In general, given the local 
offset $i'$ of a character in the  current block at level $l$, we first see if 
the current block is marked. 
If so, we map $i'$ to a position in the next level $l+1$, 
where the current block is split into two blocks of half the length:
if $i' \le b_{l+1}$, then we continue on the left child with the same offset;
otherwise, we subtract $b_{l+1}$ from $i'$ and continue on the
right child.
If, instead, $i$ is not in a marked block, we take the pointer
$\langle ptr_1,ptr_2,\delta\rangle$ stored for that block and add $\delta$ to
$i'$. If the result is $i' \le  b_l$, then we continue in the
node $ptr_1$ with offset $i'$; otherwise, we continue in $ptr_2$ with the
offset $i' -  b_l$. In both cases, the new node is marked, so
we proceed as on marked blocks in order to move to the next level in 
constant time.
The total time to extract a symbol is then $\bigO(\log(n/\gamma))$.

A substring of length $\ell$ can thus be extracted in time $\bigO(\ell
\log(n/\gamma))$, which will be sufficient to obtain the search time complexity
of Theorem~\ref{thm:main}.
It is possible to augment $\Gamma$-trees to match the complexity 
$\bigO(\log(n/\gamma)+\ell/\log_\sigma n)$ obtained by
Kempa and Prezza \cite{kempa2018roots} on string attractors as well, though this
would have no impact on our results.

\subsection{Fingerprinting} \label{sec:finger}

We now show that the $\Gamma$-tree can be augmented to compute the Karp--Rabin fingerprint of any text substring in logarithmic time. 


\bigskip
\begin{lemma}\label{lemma:KR}
Let $S[1..n]$ have an attractor of size $\gamma$ and $\phi$ a Karp--Rabin
fingerprint function. Then we can store a data structure of size 
$\bigO(\gamma\log(n/\gamma))$ words supporting the computation of $\phi$ on
any substring of $S$ in $\bigO(\log(n/\gamma))$ time. 
\end{lemma}
\begin{proof}
	We augment our $\Gamma$-tree of $S$ as follows. At level $0$, we store the Karp--Rabin fingerprints of all the text prefixes ending at positions $i\cdot n/\gamma$, for $i=1, \dots, \gamma$. 
	At levels $l>0$, we store the fingerprints of all explicit blocks. 		

We first show that we can reduce the problem to that of computing the fingerprints of two prefixes of explicit blocks. Then, we show how to solve the sub-problem of computing fingerprints of prefixes of explicit blocks.
	
	Let $S[i..j]$ be the substring of which we wish to compute the
fingerprint $\phi(S[i..j])$. Note that $\phi(S[i..j])$ can be
computed in constant time from $\phi(S[1..i-1])$ and $\phi(S[1..j])$ so we can
assume, without loss of generality, that $i=1$ (i.e., the substring is a prefix of $S$). Then, at level 0 the substring spans a sequence $B_1\cdots B_t$ of blocks followed by a prefix $C$ of block $B_{t+1}$ (the sequence of blocks or $C$ could be empty). The fingerprint of $B_1\cdots B_t$ is explicitly stored, so the problem reduces to that of computing the fingerprint of $C$. 
	
	We now show how to compute the fingerprint of a prefix of  an explicit block (at any level) in $\bigO(\log(n/\gamma))$ time. We distinguish two cases.
	
	(A) We wish to compute the fingerprint of $B[1..k]$, for some $k\leq b_l$, and $B$ is a marked block at level $l$. Let $\Bleft$ and $\Bright$ be the children of $B$ at level $l+1$.
	Then, the problem reduces to either (i) computing the fingerprint of
$\Bleft[1..k]$ if $k\leq b_{l+1}$, or combining the fingerprints of $\Bleft$
(which is stored) and $\Bright[1..k-b_{l+1}]$. In both sub-cases, the problem reduces to that of computing the fingerprint of the prefix of a block at level $l+1$, which is explicit since $B$ is marked. 
	
	(B) We wish to compute the fingerprint of $B[1..k]$, for some $k\leq b_l$, but $B$ is an unmarked explicit block. Then, $B$ is linked (through a $\Gamma$-tree pointer) to an occurrence in the same level spanning at most two blocks, both of which are marked. If the occurrence of $B$ spans only one marked block $B'$ at level $l$, then $B[1..b_l] = B'[1..b_l]$ and we are back in case (A). Otherwise, the occurrence of $B$ spans two marked blocks $B'$ and $B''$ at level $l$: $B[1..b_l] = B'[i..b_l]B''[1..i-1]$, with $i\leq b_l$. For each pointer of this kind in the $\Gamma$-tree, we store the fingerprint of $B'[i..b_l]$. We consider two sub-cases. (B.1) If $k\geq b_l-i+1$, then $B[1..k] = B'[i..b_l]B''[1..k-(b_l-i+1)]$. Since we store the fingerprint of $B'[i..b_l]$, the problem reduces again to that of computing the fingerprint of the prefix $B''[1..k-(b_l-i+1)]$ of a marked (explicit) block. 
	(B.2) If $k < b_l-i+1$, then $B[1..k] = B'[i..i+k-1]$. Although this
is not a prefix nor a suffix of a block, note that $B[1..k]B'[i+k..b_l] =  B'[i..i+k-1]B'[i+k..b_l] = B'[i..b_l]$. It follows that we can retrieve the fingerprint of $B[1..k]$ in constant time using the fingerprints of $B'[i+k..b_l]$ and $B'[i..b_l]$. The latter value is explicitly stored. The former is the fingerprint of the suffix of an explicit (marked) block. 
	In this case, note that the fingerprint of a block's suffix can be retrieved from the fingerprint of the block and the fingerprint of a block's prefix, so we are back to the problem of computing the fingerprint of an explicit block's prefix.

	To sum up, computing a prefix of an explicit block at level $l$
reduces to the problem of computing a prefix of an explicit block at level
$l+1$ (plus a constant number of arithmetic operations to combine values). In
the worst case, we navigate down to the leaves, where fingerprints of single characters can be computed in constant time. Combining this procedure into our main algorithm, we obtain the claimed running time of $\bigO(\log(n/\gamma))$. 
\end{proof}

\section{A Universal Self-Index}

Our self-index structure builds on the $\Gamma$-tree of $S$. It is formed by 
two main components: the first finds all the pattern occurrences that cross 
explicit block boundaries, whereas the second finds the occurrences that are
completely inside unmarked blocks. 

\bigskip
\begin{lemma} \label{lem:explicit}
Any substring $S[i..j]$ of length at least $2$ either overlaps two consecutive 
explicit blocks or is completely inside an unmarked block.
\end{lemma}
\begin{proof}
The leaves of the $\Gamma$-tree partition $S$ into a sequence of explicit
blocks: $\gamma$ of those are attractor positions and the other $w-\gamma$
are unmarked blocks. Clearly, if $S[i..j]$ is not completely inside an
unmarked block, it must cross a boundary between two explicit blocks.
\end{proof}

We exploit the lemma in the following way. We will define an occurrence of $P$
as {\em primary} if it overlaps two consecutive explicit blocks. 
The occurrences that are completely contained in an unmarked block are {\em
secondary}. By the lemma, every occurrence of $P$ is either primary or 
secondary. We will use a data structure to find the primary occurrences and 
another one to detect the secondary ones. The primary occurrences are found by
exploiting the fact that a prefix of $P$ matches at the end of an explicit 
block and the remaining suffix of $P$ matches the text that follows.
Secondary occurrences,
instead, are found by detecting primary or other secondary occurrences within 
the area where an unmarked block points. 

We note that this idea is a variant of the classical one \cite{KU96} used in 
all indexes based on LZ77 and CFGs. Now we show that the principle 
can indeed be applied on attractors, which is the general concept underlying 
all those compression methods (and others where no indexes exist yet), and 
therefore unifies all those particular techniques.

\subsection{Primary Occurrences}
\label{sec:primary}

We describe the data structures and algorithms used to find the primary
occurrences. Overall, they require $\bigO(w)$ space and
find the $occ_p$ primary occurrences in time $\bigO(m\log(mn/\gamma) + 
occ_p \log^\epsilon w)$, for any constant $\epsilon>0$.

\paragraph{Data Structures.}

The leaves of the $\Gamma$-tree partition
$S$ into explicit blocks. The partition is given by the starting positions
$1=p_1 < \ldots < p_w \le n$ of the leaves.
By Lemma~\ref{lem:explicit}, every primary occurrence contains some substring
$S[p_i-1..p_i]$ for $1 < i \le w$. 

If we find the occurrences considering only their leftmost covered position
of the form 
$p_i-1$, we also ensure that each primary occurrence is found once. Thus, we
will find primary occurrences as a prefix of $P$ appearing at the end of
some $S[p_{i-1}..p_i-1]$ followed by the corresponding suffix of $P$ appearing
at the beginning of $S[p_i..n]$.
For this sake, we define the set of pairs $\mathcal{B} = 
\{ \langle S[p_i..n], S[p_{i-1}..p_i-1]^{rev}\rangle,~1 < i \le w \}$, where
$S[p_{i-1}..p_i-1]^{rev}$ means $S[p_{i-1}..p_i-1]$ read backwards, and
form multisets $\mathcal{X}$ and $\mathcal{Y}$ with the left and
right components of $\mathcal{B}$, respectively.

We then lexicographically sort $\mathcal{X}$ and $\mathcal{Y}$, to obtain the strings $X_1, X_2, 
\ldots$ and $Y_1, Y_2, \ldots$. All the occurrences ending with a certain prefix
of $P$ will form a contiguous range in the sorted multiset $\mathcal{Y}$, whereas all 
those starting with a certain suffix of $P$ will form a contiguous range in 
the sorted 
multiset $\mathcal{X}$. Each primary occurrence of $P$ will then correspond to a pair 
$\langle X_x,Y_y \rangle \in \mathcal{B}$ where both $X_x$ and $Y_y$ belong to 
their range.

Our structure to find the primary occurrences is a two-dimensional discrete 
grid $G$ storing one point $(x,y)$ for each pair $\langle X_x,Y_y \rangle
\in \mathcal{B}$.
The grid $G$ is of size $(w-1) \times (w-1)$.
We represent $G$ using a two-dimensional range search data structure 
requiring $\bigO(w)$ space \cite{CLP11} that reports 
the $t$ points lying inside any rectangle of the grid in time 
$\bigO((t+1)\log^\epsilon w)$, for any constant $\epsilon>0$.
We also store an array $T[1..w-1]$ that, for each point $(x,y)$ in $G$,
where $X_x = S[p_i..n]$, stores $T[y] = p_i$, that is, where $X_x$ starts in
$S$.

\paragraph{Queries.}
To search for a pattern $P[1..m]$, we first find its primary occurrences using
$G$ as follows. For each partition $P_< = P[1..k]$ and $P_> = P[k+1..m]$, for
$1 \le k < m$, we search $\mathcal{Y}$ for $P_<^{rev}$ and $\mathcal{X}$ for 
$P_>$. For each identified range $[x_1,x_2] \times [y_1,y_2]$, we extract all
the $t$ corresponding primary occurrences $(x,y)$ in time 
$\bigO((t+1)\log^\epsilon w)$ with our range search data structure. Then we
report a primary occurrence starting at $T[y]-k$ for each such point $(x,y)$.
Over the $m$ intervals, this adds up to $\bigO((m+occ_p)\log^\epsilon w) =
\bigO((m+occ_p)\log^\epsilon(\gamma\log(n/\gamma))$.

We obtain the $m-1$ ranges in the multiset $\mathcal{X}$ in overall time 
$\bigO(m\log(mn/\gamma))$, by using the fingerprint-based technique of Gagie 
et al.~\cite{GGKNP14} applied to the z-fast trie of Belazzougui et al.~\cite{belazzougui2010fast} (we use a lemma from Gagie et al.~\cite{GNP18} where the overall result is stated in cleaner form). We use an analogous structure to obtain the 
ranges in $\mathcal{Y}$ of the suffixes of the reversed pattern.

\bigskip
\begin{lemma}[{adapted from \cite[Lem.\ 5.2]{GNP18}}] \label{lem:ztrie}
Let $S[1..n]$ be a string on alphabet $[1..\sigma]$, $\mathcal{X}$ be a sorted 
set of suffixes of $S$, and $\phi$ a Karp--Rabin fingerprint function.
If one can extract a substring of length $\ell$ from $S$ in time $f_e(\ell)$ 
and compute $\phi$ on it in time $f_h(\ell)$, then one 
can build a data structure of size $\bigO(|\mathcal{X}|)$ that obtains
the lexicographic ranges in $\mathcal{X}$ of the $m-1$ suffixes of a given
pattern $P$ in worst-case time $\bigO(m\log(\sigma)/\omega + m(f_h(m)+\log
m)+f_e(m))$~--- provided that $\phi$ is collision-free among substrings of $S$ whose lengths are powers of two.
\end{lemma}

In Sections~\ref{sec:extract} and \ref{sec:finger} we show how to
extract in time $f_e(\ell) = \bigO(\ell \lg(n/\gamma))$ and how to compute a
fingerprint in time $f_h(\ell) = \bigO(\lg(n/\gamma))$, respectively.  
In Section \ref{sec:building fingerprints} we show that a Karp--Rabin function that is collision-free  among substrings whose lengths are powers of two can be efficiently found. Together, these results 
show that we can find all the ranges in $\mathcal{X}$ and $\mathcal{Y}$ in time
$\bigO(m\lg(\sigma)/\omega + m(\lg(n/\gamma)+\lg m)+m\lg(n/\gamma)) =
\bigO(m\lg(mn/\gamma))$.

Patterns $P$ of length $m=1$ can be handled as $P[1]*$, where $*$ stands for
any character. Thus, we take $[x_1,x_2]=[1,w]$ and carry out the search as a 
normal pattern of length $m=2$. To make this work also for the last position 
in $S$, we assume as usual that $S$ is terminated by a special symbol \$ that cannot appear in search patterns $P$. Alternatively, we can store a list of the
marked leaves where each alphabet symbol appears, and take those as the primary
occurrences. A simple variant of Lemma~\ref{lem:explicit} shows that 
occurrences of length 1 are either attractor positions or are inside an 
unmarked block, so the normal mechanism to find secondary occurrences from 
this set works correctly. Since there are in total $\gamma$ such leaves, the 
space for these lists is $\bigO(\sigma+\gamma)=\bigO(\gamma)$.

\subsection{Secondary Occurrences}

We now describe the data structures and algorithms to find the secondary
occurrences. They require $\bigO(w)$ space and
find the $occ_s$ secondary occurrences in time 
$\bigO((occ_p+occ_s) \log\log_\omega (n/\gamma))$.

\paragraph{Data structures.}
To track the secondary occurrences, let us call {\em target} and {\em source} 
the text areas $S[i..i']$ and $S[j'..j'']$, respectively, of an unmarked
block and its pointer, so that there is some $j \in \Gamma$ contained in 
$S[j'..j'']$ (if the blocks are at level $l$, then $i'=i+b_l-1$ and $j''=j'+
b_l-1$). Let $S[pos..pos+m-1]$ be an 
occurrence we have already found (using the grid $G$, initially). Our aim is
to find all the sources that contain $S[pos..pos+m-1]$, since their 
corresponding targets then contain other occurrences of $P$.

To this aim, we store the sources of all levels in an array $R[1..w-\gamma]$,
with fields $j'$ and $j''$, ordered by starting positions $R[k].j'$. 
We build a predecessor search structure on the fields $R[k].j'$, and a range 
maximum query (RMQ) data structure on the fields $R[k].j''$, able to find the 
maximum endpoint $j''$ in any given range of $R$. While a predecessor search 
using $\bigO(w)$ space requires $\bigO(\log\log_\omega (n/w))$ time on an
$\omega$-bit-word machine \cite[Sec. 1.3.2]{PT06}, the RMQ structure operates in constant 
time using just $\bigO(w)$ bits \cite{FH11}.

\paragraph{Queries.}
Let $S[pos..pos+m-1]$ be a primary occurrence found.
A predecessor search for $pos$ gives us the rightmost position $r$ where the 
sources start at $R[r].j' \le pos$. An RMQ on $R[1..r]$ then finds the position
$k$ of the source with the rightmost endpoint $R[k].j''$ in $R[1..r]$. If even 
$R[k].j'' < pos+m-1$, then no source covers the occurrence and we finish. If, 
instead, $R[k].j'' \ge pos+m-1$, then the source $R[k]$ covers the occurrence 
and we process its corresponding target as a secondary occurrence; in this case
we also recurse on the ranges $R[1..k-1]$ and $R[k+1..r]$ that are nonempty. 
It is easy to see that each valid secondary occurrence is identified in 
$\bigO(1)$ time (see Muthukrishnan \cite{Mut02} for an analogous process).
In addition, such secondary occurrences, $S[pos'..pos'+m-1]$, must be
recursively processed for further secondary occurrences.
A similar procedure is described for tracking the secondary occurrences in 
the LZ77-index \cite{KN13}.

The cost per secondary occurrence reported then amortizes to a predecessor 
search, $\bigO(\lg\lg_\omega(n/\gamma))$ time. This cost is also paid for each
primary occurrence, which might not yield any secondary occurrence to amortize
it. We now prove that this process is sufficient to find all the secondary
occurrences.

\bigskip
\begin{lemma}
The described algorithm reports every secondary occurrence exactly once.
\end{lemma}
\begin{proof}
We use induction on the level $l$ of the unmarked block that contains the
secondary occurrence. All the secondary occurrences of the last level are 
correctly reported once, since there are none. Now consider an occurrence 
$S[pos..pos+m-1]$ inside an unmarked block $S[i..i']$ of level $l$. This block 
is the target of a source $S[j'..j'']$ that spans 1 or 2 consecutive marked 
blocks of level $l$. Then there is another occurrence $S[pos'..pos'+m-1]$, 
with $pos' = pos-i+j'$. Note that the algorithm can report $S[pos..pos+m-1]$
only as a copy of $S[pos'..pos'+m-1]$. If $S[pos'..pos'+m-1]$ is primary, then 
$S[pos..pos+m-1]$ will be reported right after $S[pos'..pos'+m-1]$, because 
$[pos'..pos'+m-1] \subseteq [j'..j'']$, and the algorithm will map $[j'..j'']$ 
to $[i..i']$ to discover $S[pos..pos+m-1]$. Otherwise, $S[pos'..pos'+m-1]$ is 
secondary, and thus it is within one of the marked blocks at level $l$ that overlap the 
source. Moreover, it is within one of the blocks of level $l+1$ into which 
those marked blocks are split. Thus, $S[pos'..pos'+m-1]$ is within an unmarked 
block of level $>l$, which by the inductive hypothesis will be reported 
exactly once. When $S[pos'..pos'+m-1]$ is reported, the algorithm will also 
note that $[pos'..pos'+m-1] \subseteq [j'..j'']$ and will find 
$S[pos..pos+m-1]$ once as well.

If we handle patterns of length $m=1$ by taking the corresponding attractor 
positions as the primary occurrences, then there may be secondary occurrences
in the last level, but those point directly to primary occurrences (i.e., 
attractor positions), and therefore the base case of the induction holds too.
\end{proof}

\medskip
The total search cost with $occ$ primary and secondary occurrences is therefore
$\bigO(m(\log^\epsilon(\gamma\log(n/\gamma))+\log(mn/\gamma)) + 
occ (\log^\epsilon (\gamma\log(n/\gamma)) + 
\log\log_\omega(n/\gamma))) = 
\bigO(m\log n + occ (\log^\epsilon \gamma + \log\log (n/\gamma)))
= \bigO(m\log n + occ\log^\epsilon n)$, for any constant
$\epsilon>0$ defined at indexing time (the choice of $\epsilon$ affects the 
constant that accompanies the size $\bigO(\gamma\log(n/\gamma))$ of the structure $G$).

\section{Construction}

If we allow the index construction to be correct with high probability only, 
then we can build it in $\bigO(n+w\lg(n/\gamma)+w\sqrt{\lg w})$ time and
$\bigO(w)$ space (plus read-only access to $S$), using a Monte Carlo method. 
Since $w = \bigO(\gamma \lg(n /\gamma))$ is the number of leaves in the 
$\Gamma$-tree and $\gamma(\log(n/\gamma))^{\bigO(1)} \subseteq \bigO(n)$, the 
time can be written as $\bigO(n+w\sqrt{\log \gamma})$.
In order to ensure a correct index, a Las Vegas method requires 
$\bigO(n\log n)$ time in expectation (and w.h.p.) and $\bigO(n)$ space. 

%
%
%
%

\subsection{Building the $\Gamma$-tree}

Given the attractor $\Gamma$, we can build the index data structure as follows.
At each level $l$, we create an Aho--Corasick automaton \cite{AC75} on the
unmarked blocks at this level (i.e., those at distance $\ge b_l$ from any
attractor), and use it to scan the areas $S[j-b_l+1..j+b_l-1]$ around all the
attractor elements $j \in \Gamma$ in order to find a proper pointer for each 
of those unmarked blocks. This takes $\bigO(n)$ time per level. Since the areas
around each of the $\gamma$ attractor positions are scanned at each level but
they have exponentially decreasing lengths, the scanning time adds up to
$\bigO(\gamma\frac{n}{\gamma} + \gamma\frac{n}{2\gamma} + 
\gamma\frac{n}{4\gamma} + \cdots) = \bigO(n)$. 

As for preprocessing time, each unmarked block is preprocessed only once, and 
they add up to $\bigO(n)$ symbols. The preprocessing can thus be done in time 
$\bigO(n)$ \cite{DL06}. To be able to scan in linear time,
we can build deterministic dictionaries on the edges outgoing from 
each node, in time $\bigO(n(\log\log\sigma)^2)$ \cite{Ruz08}. 

In total, we can build the $\Gamma$-tree in $\bigO(n(\log\log\sigma)^2)$ 
deterministic time and $\bigO(n)$ space.

To reduce the space to $\bigO(w)$, instead of inserting the unmarked blocks 
into an Aho--Corasick automaton, we compute their Karp--Rabin fingerprints,
store them in a hash table, and scan the areas $S[j-b_l+1..j+b_l-1]$ around
attractor elements $j$. This finds the correct sources for all the unmarked 
blocks w.h.p. Indeed, if we verify the potential collisions, 
the result is always correct within $\bigO(n)$ expected time (further, this
time holds w.h.p.).

\subsection{Building the Fingerprints}\label{sec:building fingerprints}

Building the structures for Lemma~\ref{lemma:KR} requires
(i)  computing the fingerprint of every text prefix ending at block boundaries 
($\bigO(n)$ time and $\bigO(w)$ space in addition to $S$), 
(ii) computing the fingerprint of every explicit block ($\bigO(w)$ time and 
$\bigO(w)$ space starting from the leaves and combining results up to the root),
(iii) for each unmarked explicit block $B$, computing the fingerprint of a
string of length at most $|B|$ (i.e., the fingerprint of $B'[i..b_l]$; see case (B) of Lemma \ref{lemma:KR}). Since 
unmarked blocks do not have children, each text character is seen at 
most once while computing these fingerprints, which implies that these values can also be computed
in $\bigO(n)$ time and $\bigO(w)$ space in addition to $S$. 

This process, however, does not include finding
a collision-free Karp--Rabin hash function. As a result, the fingerprinting
is correct w.h.p.\ only. 
We can use the de-randomization procedure of Bille et 
al.~\cite{bille2015longest}, which guarantees to find~--- in $\bigO(n\log n)$ 
expected time\footnote{The time is also $\bigO(n\log n)$ w.h.p., not only in
expectation.} and $\bigO(n)$ words of space~--- a Karp--Rabin hash function that is collision-free among substrings of $S$ whose lengths are powers of two. This is sufficient to deterministically check the equality of substrings\footnote{If the length $\ell$ of the two substrings is not a power of two, then we compare their prefixes and suffixes whose length is the largest power of two smaller than $\ell$.} in the z-fast trie used in the technique~\cite[Lem.~5.2]{GNP18} that we use to quickly find ranges of pattern suffixes/prefixes (in our Section~\ref{sec:primary}).

\subsection{Building the Multisets $\mathcal{X}$ and $\mathcal{Y}$}

To build the multisets $\mathcal{X}$ and $\mathcal{Y}$ for the primary occurrences, we can build
the suffix arrays \cite{MM93} of $S$ and its reverse, $S^{rev}$. This requires 
$\bigO(n)$ deterministic time and space \cite{KSB06}. Then we can scan those 
suffix arrays to enumerate $\mathcal{X}$ and $\mathcal{Y}$ in the lexicographic order.

To sort $\mathcal{X}$ and $\mathcal{Y}$ within $\bigO(w)$ space, we can build 
instead a sparse suffix tree on the $w$ positions of $\mathcal{X}$ or 
$\mathcal{Y}$ in $S$. This can be done in expected time (and w.h.p.)
$\bigO(n\sqrt{\log w})$ and $\bigO(w)$ space \cite{GK17}. If we
aim to build the suffix array correctly w.h.p.\ only, then the 
time drops to $\bigO(n)$.

We must then build the z-fast trie \cite[Thm.~5]{belazzougui2010fast} on the
sets $\mathcal{X}$ and $\mathcal{Y}$. Since we can use any space in
$\bigO(w)$, we opt for a simpler variant described by Kempa and Kosolobov
\cite[Lem.~5]{KK17}, which is easier to build. Theirs is a compact trie that
stores, for each node $v$ representing the string $v.str$ and with parent node
$v.par$: (i) the length $|v.str|$, (ii) a dictionary mapping each character 
$c$ to the child node $v'$ of $v$ such that $v'.str[|v.str|+1]=c$ (if such a 
child exists), and
(iii) the (non-extended) fingerprint of $v.str[1..k_v]$,
where $k_v$ is the \emph{two-fattest number} in the range $[|v.par.str|+1..|v.str|]$, that is, the number in that range whose binary representation has the largest number of trailing zeros. 
The trie also requires a global ``navigating'' hash table that maps the
$\bigO(w)$ pairs $(k_v,\hat\phi(v.str[1..k_v]))$ to their 
corresponding node $v$. 

If $p$ prefixes some string in $\mathcal{X}$ (resp.\ $\mathcal{Y}$) and the
fingerprint function $\hat\phi$ is collision-free among equal-length text 
substrings, then their so-called \emph{fat binary search} procedure finds~--- 
in time $\bigO(\log|p|)$~--- the highest node $v$ such that $v.str$ is prefixed
by a search pattern $p$ (assuming constant-time computation of the fingerprint 
of any substring of $p$, which can be achieved after a linear-time preprocessing
of the pattern). The range of strings of $\mathcal{X}$ (resp.\ $\mathcal{Y}$)
is then associated with the node $v$. 

Their search procedure, however, has an anomaly \cite[Lem.~5]{KK17}: in some
cases it might return a child of $v$ instead of $v$. We can fix this problem
(and always return the correct $v$) by also storing, for every internal trie 
node $v$, the fingerprint of $v.str$ and the exit character 
$v.str[|v.par.str|+1]$. Then, we know that the procedure returned $v$ if and 
only if $|v.par.str| < |p|$, the fingerprint of $v.par.str$ matches that of 
$p[1..|v.par.str|]$ and 
$v.str[|v.par.str|+1] = p[|v.par.str|+1]$. If these conditions do not hold, we 
know that the procedure returned a child of $v$ and we fix the answer by moving
to the parent of the returned node.\footnote{This fix works in general. In
their particular procedure, we do not need to check that $v.str[|v.par.str|+1] =
p[|v.par.str|+1]$ (nor to store the exit character).
Alternatively, we could directly fix their procedure by similarly storing the 
fingerprint of 
$v.str$ ($\mathsf{hash}(v.str)$, in their article) and changing line 6
of their extended version \cite{KK17} to ``{\bf if} $v.len < |pat|
~\mathbf{and}~\mathsf{hash}(v.str) = \mathsf{hash}(pat[1..v.len]) 
~\mathbf{and}~ v.map(pat[v.len+1]) \not=$ {\bf nil then} $v \leftarrow
v.map(pat[v.len+1])$;''.} 

If $p$ does not prefix any string in $\mathcal{X}$ (resp. $\mathcal{Y}$), or 
if $\hat\phi$ is not collision-free, the search (with our without our fix) 
returns an arbitrary range. We recall that, to cope with the first 
condition (which gives the z-fast trie search the name {\em weak} prefix 
search), Lemma \ref{lem:ztrie} adds a deterministic check that allows
discarding the incorrect ranges returned by the procedure. The second
condition, on the other hand, w.h.p.\ does not fail. Still, we recall that we
can ensure that our function is collision-free by checking, at construction
time, for collisions among substrings whose lengths are powers of two, as 
described in Section \ref{sec:building fingerprints}.

The construction of this z-fast trie starts from the sparse suffix tree of
$\mathcal{X}$ (or $\mathcal{Y}$), since the topology is the same. 
The dictionaries giving constant-time access to the nodes' children can be built correctly w.h.p.\ in 
time $\bigO(w)$, or be made perfect in expected time (and w.h.p.) $\bigO(w)$
\cite{Wil00}. 
Alternatively, one can use deterministic dictionaries, which can be built in worst-case
time $\bigO(w(\log\log \sigma)^2)$ \cite{Ruz08}. In all cases, the construction
space is $\bigO(w)$.
Similarly, the navigating hash can be built correctly w.h.p.\ in 
time $\bigO(w)$, or be made perfect in expected time (and w.h.p.) $\bigO(w)$. 
We could also represent the navigating hash using
deterministic dictionaries built in worst-case
time $\bigO(w(\log\log w)^2)$ \cite{Ruz08} and space $\bigO(w)$. 
The $\bigO(w)$ fingerprints $\hat\phi(v.str)$ and the hashes of the pairs $(k_v,\hat\phi(v.str[1..k_v]))$ can be computed in 
total time $\bigO(w\log(n/\gamma))$ by using Lemma~\ref{lemma:KR} to compute 
the Karp--Rabin fingerprints.

\subsection{Structures to Track Occurrences}

A variant of the grid data structure $G$ \cite{BP16}, with the same
space and time performance, can be built in $\bigO(w\sqrt{\log w})$ 
deterministic time, and $\bigO(w)$ space. The arrays $T$ and $R$
can be built in $\bigO(w)$ space and $\bigO(w\log\log w)$ deterministic
time \cite{Han04}.

The RMQ structure on $R$ requires $\bigO(w)$ deterministic construction time 
and $\bigO(w)$ bits \cite{FH11}. The predecessor data structure \cite{PT06},
however, requires perfect hashing. This can be built in $\bigO(w)$ space and 
$\bigO(w(\log\log w)^2)$ deterministic time \cite{Ruz08}, or in
$\bigO(w)$ expected time (also holding w.h.p.).

\section{Conclusions}

We have introduced the first {\em universal} self-index for repetitive text
collections. The index is based on a recent result
\cite{kempa2018roots} that unifies a large number of dictionary compression
methods into the single concept of {\em string attractor}. For each compression
method based on Lempel--Ziv, grammars, run-compressed BWT, collage systems,
macro schemes, etc., it is easy to identify an attractor set of the
same asymptotic size obtained by the compression method, say $\gamma$. Thus, our construction 
automatically yields a self-index for each of those compression methods, within
$\bigO(\gamma\log(n/\gamma))$ space. No structure is known of size 
$o(\gamma\log(n/\gamma))$ able to efficiently extract a substring from the 
compressed text (say, within $\bigO(\mathrm{polylog}\; n)$ time per symbol), 
and thus $\bigO(\gamma\log(n/\gamma))$ is the minimum space we could hope for 
an efficient self-index with the current state of the art.

Indeed, no self-index of size $o(\gamma\log(n/\gamma))$ is known: the smallest ones 
use $\bigO(z)$~\cite{KN13} and $\bigO(r)$~\cite{GNP18} space, respectively, yet there are text families where $z,r = \Omega(b^*\log n) = \Omega(\gamma^*\log n)$ \cite{GNPlatin18} (where $b^*$ and $\gamma^*$ denote the smallest macro scheme and attractor sizes, respectively).
The most time-efficient self-indexes use $\Theta(r\log(n/r))$ and $\Omega(z\log(n/z))$
space, which is asymptotically equal or larger than ours. The search time of 
our index, $\bigO(m\log n + occ \log^\epsilon n)$ for any constant $\epsilon>0$,
is close to that of the fastest of those self-indexes, which were
developed for a specific compression format (see \cite[Table 1]{GNP18}).
Moreover, our construction provides a self-index for compression methods for
which no such structure existed before, such as collage systems and
macro schemes. Those can provide smaller attractor sets than
the ones derived from the more popular compression methods.

We can improve the search time of our index by using slightly more space. Our current
bottleneck in the per-occurrence query time is the grid data structure $G$,
which uses $\bigO(w)$ space and returns each occurrence in time 
$\bigO(\log^\epsilon w)$. Instead, a grid structure \cite{CLP11} using 
$\bigO(w\log\log w)=\bigO(\gamma\log(n/\gamma)\log\log n)$ space obtains the $occ_p$
primary occurrences in time $\bigO((occ_p+1)\log\log w)$. This slightly larger 
version of our index can then search in time $\bigO(m\log n + occ
\log\log n)$. This complexity is close to that of some larger indexes in the
literature for repetitive string collections (see \cite[Table 1]{GNP18}).

A number of avenues for future work are open, including supporting more complex
pattern matching, handling dynamic collections of texts, supporting document
retrieval, and implementing a practical version of the index. Any advance in
this direction will then translate into all of the existing indexes for
repetitive text collections.

\section*{Acknowledgements}

We thank the reviewers for their outstanding job in improving our presentation
(and even our bounds).

\bibliographystyle{plain}
\bibliography{universal-indexing}

\end{document}